\documentclass[review]{elsarticle}
\usepackage{amsfonts,color,morefloats,pslatex}
\usepackage{amssymb,amsthm, amsmath,latexsym}

\newtheorem{theorem}{Theorem}
\newtheorem{lemma}[theorem]{Lemma}
\newtheorem{proposition}[theorem]{Proposition}

\newtheorem{conj}{Conjecture}

\newcommand{\tr}{{\mathrm{Tr}}}

\newcommand{\gf}{{\mathrm{GF}}}

\newcommand{\GA}{{\mathrm{GA}}}

\newcommand{\bS}{{\mathbb{S}}}

\newcommand{\cP}{{\mathcal{P}}}
\newcommand{\cB}{{\mathcal{B}}}

\newcommand{\bD}{{\mathbb{D}}}

\newcommand{\Segre}{ {{\mathrm{Segre}}} }
\newcommand{\Glynnone}{ {{\mathrm{Glynni}}} }
\newcommand{\Glynntwo}{ {{\mathrm{Glynnii}}} }

\begin{document}

\begin{frontmatter}



\title{Infinite families of  $3$-designs from APN functions
\tnotetext[fn1]{The research of C. Tang was supported by National Natural Science Foundation of China (Grant No.
11871058) and China West Normal University (14E013, CXTD2014-4 and the Meritocracy Research
Funds).}
}

\author{Chunming Tang}
\ead{tangchunmingmath@163.com}

\address{School of Mathematics and Information, China West Normal University, Nanchong, Sichuan,  637002, China, and
  Department of Mathematics, The Hong Kong University of Science and Technology, Clear Water Bay, Kowloon, Hong Kong}

\begin{abstract}
Combinatorial $t$-designs have nice applications in coding theory, finite geometries
and several engineering areas.
The objective of this paper is to study how to obtain $3$-designs with $2$-transitive permutation groups.  The incidence structure formed by the orbits of a base block under the action of the general
affine groups, which are $2$-transitive, is   considered. A characterization of  such   incidence structure  to be a  $3$-design  is presented, and a sufficient condition for the stabilizer of a base block to be trivial  is given.
With these general results, infinite families of $3$-designs are constructed by employing APN functions.
Some $3$-designs presented in this paper give rise to self-dual binary codes or linear codes with optimal or best parameters known. Several conjectures on $3$-designs and binary codes are also presented.
\end{abstract}


\end{frontmatter}

\section{Introduction}

Let $\cP$ be a set of $v \ge 1$ elements, and let $\cB$ be a set of $k$-subsets of $\cP$, where $k$ is
a positive integer with $1 \leq k \leq v$. Let $t$ be a positive integer with $t \leq k$. The pair $\bD = (\cP, \cB)$
is called an \emph{incidence structure}.
The incidence structure
$\bD = (\cP, \cB)$  is said     to be a $t$-$(v, k, \lambda)$ {\em design\index{design}}, or simply {\em $t$-design\index{$t$-design}}, if every $t$-subset of $\cP$ is contained in exactly $\lambda$ elements of
$\cB$. 

Let $q$ be a  power of $2$ and $\mathrm{GF}(q)$ be the finite field of order $q$.
The general affine group $\mathrm{GA}_1(q)$ of degree one consists of all the following permutations of the finite field $\mathrm{GF}(q)$:
\[\pi_{a,b}(x)=ax+b,\]
where $(a,b)\in \mathrm{GF}(q)^*\times \mathrm{GF}(q)$.
Let $B$ be a $k$-subset of   $\mathrm{GF}(q)$ and $\pi(B)=\{\pi(x): x\in B\}$, where
  $\pi \in \mathrm{GA}_1(q)$.
The orbit of $B$ under the action of $\mathrm{GA}_1(q)$ is $\mathrm{GA}_1(q)(B)=\{\pi(B): \pi \in \mathrm{GA}_1(q)\}$,
and the stabilizer of $B$ under the action of $\mathrm{GA}_1(q)$ is $\mathrm{GA}_1(q)_B=\{\pi\in \mathrm{GA}_1(q): \pi(B)=B \}$.
The incidence structure $\bS(B):=(\gf(q), \mathrm{GA}_1(q)(B))$ may be a $t$-$(q, k, \lambda)$ design for some $\lambda$,
where $\gf(q)$ is the point set, and the incidence relation is the set membership. In this case, we say that the base block  $B$ \emph{supports} a $t$-design and $(\mathrm{GF}(q), \mathrm{GA}_1(q)(B))$
is called the \emph{orbit design} of $B$.

The following theorem shows that the
incidence structure $(\gf(q), \mathrm{GA}_1(q)(B))$ is always a $2$-design (see \cite[Proposition 4.6]{BJL} or \cite{LD}).
\begin{theorem}
Let $\mathcal P=\mathrm{GF}(q)$ and $\mathcal B=\mathrm{GA}_1(q)(B)$,
where $B$ is any $k$-subset of $\mathrm{GF}(q)$ with $k\geq 2$. Then $(\mathcal P, \mathcal B)$
is a $2$-$(q,k,\lambda)$ design, where
\[\lambda= \frac{k(k-1)}{|\mathrm{GA}_1(q)_B|}.\]

\end{theorem}

The main motivation of this paper is to study how to choose a base block
$B\subseteq  \mathrm{GF}(q)$ properly so that
$(\mathrm{GF}(q), \mathrm{GA}_1(q)(B))$ is a $3$-design. We first  give a characterization of $3$-designs $(\gf(q), \GA_1(q)(B))$ by means of
their characteristic functions, the Walsh transforms of their characteristic functions, and  the number of solutions of some
special equations associated with the base block $B$. Next, we present a sufficient condition for the stabilizer of
a base block under the action of the general affine group to be trivial.
Finally, we introduce two constructions of $3$-designs from APN functions.
Specifically,  using the first construction,
we obtain $\frac{\phi(n)}{2}$ different $3$-designs with parameters $(2^n, 2^{n-1},2^{n-3}(2^n-4))$
from Kassami APN functions. Magma programs show that
many new $3$-designs can be obtained from the second construction. Finally, we show that
some of the linear codes from the $3$-designs of this paper
are optimal or self-dual.

The rest of this paper is arranged as follows. Section \ref{sec:characterization} gives a characterization of those base blocks supporting
$3$-designs. Section \ref{sec:stabi} presents a sufficient condition for the stabilizer of a base block to be
trivial. Section \ref{sec-construct}  presents two constructions of $3$-designs from APN functions.
Section \ref{sec:contribution} concludes this paper and makes some remarks.

\section{The characterization of the base blocks supporting  $3$-designs}\label{sec:characterization}

Let $n$ be a positive  integer and $q=2^n$.
For any Boolean function $f$ from $\mathrm{GF}(2^n)$ to $\mathrm{GF}(2)$, the \emph{Walsh transform} of $f$ at $ \mu \in  \mathrm{GF}(2^n)$  is defined as
\begin{align*}
\hat{f}(\mu )=\sum_{x\in \mathrm{GF}(2^n)} (-1)^{f(x)+ \mathrm{Tr}\left(\mu x\right) },
\end{align*}
where $ \mathrm{Tr}(\cdot)$   is the absolute trace function from $\mathrm{GF}(2^n)$ to $\mathrm{GF}(2)$.
All the values $\hat{f}(\mu )$ are also called the  \emph{Walsh coefficients} of $f$.
The Boolean function
$f$ is said to be \emph{semi-bent} if $\{\hat{f}(\mu ): \mu \in \mathrm{GF}(2^n)\}=\{0, \pm 2^{\frac{n+1}{2}}\}$.
Hence semi-bent functions over $\mathrm{GF}(2^n)$ exist only for odd $n$.

Let $B$ be a subset of $\mathrm{GF}(q)$. Then, the characteristic  function $f_{B}(x)$ of $B$ is given by
\begin{align*}
f_{B}(x)=
\left\{
  \begin{array}{ll}
1, & x \in B,\\
 0, & \text{otherwise}.
  \end{array}
\right.
\end{align*}

Some results on characteristic functions are given in the following lemmas.

\begin{lemma}\label{lem:one-two=0}
Let $a, b$  be two distinct elements in $ \mathrm{GF}(q)$.  Let $B$ be a $k$-subset of  $\mathrm{GF}(q)$.
Then

(1) $\hat{f}_B(0)=q-2k$.

(2) $\sum_{x ,y\in \mathrm{GF}(q)} (-1)^{f_{B}(ax+y)}= q(q-2k)$.

(3) $\sum_{x, y \in \mathrm{GF}(q)} (-1)^{f_{B}(a x+y)+f_{B}(b x+y)}=(q-2k)^2$.
\end{lemma}
\begin{proof}
By the definition of Walsh transform, one has
\begin{align*}
\hat{f}_B(0)=& \sum_{x\in \mathrm{GF}(q)} (-1)^{f_B(x)}\\
=& q-|B|-|B|\\
=& q-2k.
\end{align*}
The conclusion of Part (1) then follows.

For $a \in \mathrm{GF}(q)$, one gets
\begin{align*}
\sum_{x ,y\in \mathrm{GF}(q)} (-1)^{f_{B}(ax+y)}=& \sum_{x \in \mathrm{GF}(q)}  \sum_{y \in \mathrm{GF}(q)}  (-1)^{f_{B}(ax+y)}\\
=& \sum_{x \in \mathrm{GF}(q)}  \sum_{y \in \mathrm{GF}(q)}  (-1)^{f_{B}(y)}\\
=& q\hat{f}_B(0)\\
=& q(q-2k).
\end{align*}

Since $a\neq b$,  $(x,y) \mapsto (ax+y, bx+y)$ is a bijection over $\mathrm{GF}(q)^2$. Then
\begin{align*}
\sum_{x, y \in \mathrm{GF}(q)} (-1)^{f_{B}(a x+y)+f_{B}(b x+y)}=& \sum_{x, y \in \mathrm{GF}(q)} (-1)^{f_{B}(x)+f_{B}(y)}\\
=& \left (\sum_{x \in \mathrm{GF}(q)} (-1)^{f_{B}(x)} \right )^2\\
=& (\hat{f}_B(0))^2\\
=& (q-2k)^2.
\end{align*}
This completes the proof.
\end{proof}

\begin{lemma}\label{lem:add-mult}
Let $a,b \in \mathrm{GF}(q)$ and $E$ be any $k$-subset of $\mathrm{GF}(q)$. Then
\[ \sum_{x,y\in \mathrm{GF}(q)} (-1)^{f_{E}(x)+f_{E}(y)+f_{E}(ax+by)}   =\frac{1}{q}\sum_{\alpha\in \mathrm{GF}(q)} \hat{f}_{E}(a \alpha) \hat{f}_{E}(b \alpha) \hat{f}_{E}( \alpha).\]
\end{lemma}
\begin{proof}
Let $S=\sum_{x,y\in \mathrm{GF}(q)} (-1)^{f_{E}(x)+ f_{E}(y)+f_{E}(ax+by)  }$.
Then
\begin{align*}
S=&\frac{1}{q} \sum_{x,y,z\in \mathrm{GF}(q)} (-1)^{f_{E}(x)+ f_{E}(y)+f_{E}(z)}\sum_{\alpha\in \mathrm{GF}(q)} (-1)^{\mathrm{Tr}(\alpha(ax+by+z))}\\
=&\frac{1}{q} \sum_{\alpha\in \mathrm{GF}(q)} \hat{f}_{E}(a \alpha) \hat{f}_{E}(b \alpha) \hat{f}_{E}( \alpha).
\end{align*}
This completes the proof.
\end{proof}

Let $E$ be any subset of $\mathrm{GF}(q)$ and $a,b,c\in \mathrm{GF}(q)$. Define
\begin{align*}
N_E(a,b,c)=|\{ax+by+cz=0: x,y,z\in E\}|.
\end{align*}
The following lemma gives the relation between characteristic functions and $N_E(a,b,c)$.
\begin{lemma}\label{lem:prodof3spec}
Let $a,b, c \in \mathrm{GF}(q)^*$ and $E$ be any $k$-subset of $\mathrm{GF}(q)$. Then
\[\frac{1}{q}\sum_{\alpha\in \mathrm{GF}(q)} \hat{f}_{E}(a \alpha) \hat{f}_{E}(b \alpha) \hat{f}_{E}( c \alpha) =q^2-6kq+12k^2- 8 N_E(a, b, c).\]
\end{lemma}

\begin{proof}
Let $\alpha\neq 0$, then
\begin{align*}
\hat{f}_E(\alpha)=& \sum_{x\in \mathrm{GF}(q)}  (-1)^{f_E(x)} (-1)^{\mathrm{Tr}(\alpha x)}\\
=& \sum_{x\in \mathrm{GF}(q)}  \left ( (-1)^{f_E(x)} -1 \right ) (-1)^{\mathrm{Tr}(\alpha x)} +\sum_{x\in \mathrm{GF}(q)}(-1)^{\mathrm{Tr}(\alpha x)}\\
=& -2 \sum_{x\in E} (-1)^{\mathrm{Tr}(\alpha x)} +\sum_{x\in \mathrm{GF}(q)}(-1)^{\mathrm{Tr}(\alpha x)}\\
=& -2 \sum_{x\in E} (-1)^{\mathrm{Tr}(\alpha x)}.
\end{align*}
Denote $\frac{1}{q}\sum_{\alpha\in \mathrm{GF}(q)} \hat{f}_{E}(a \alpha) \hat{f}_{E}(b \alpha) \hat{f}_{E}( c \alpha)$ by $S$, then

\begin{align*}
S=& \frac{1}{q} \sum_{\alpha\in \mathrm{GF}(q)^*} \hat{f}_{E}(a \alpha) \hat{f}_{E}(b \alpha) \hat{f}_{E}( c \alpha) +  \frac{1}{q}  \hat{f}_{E}(0)^3\\
=& \frac{-8}{q} \sum_{\alpha\in \mathrm{GF}(q)^*} \left ( \sum_{x\in E} (-1)^{\mathrm{Tr}(a \alpha x)} \right ) \left ( \sum_{x\in E} (-1)^{\mathrm{Tr}(b \alpha x)} \right )
\left ( \sum_{x\in E} (-1)^{\mathrm{Tr}(c \alpha x)} \right )+  \frac{1}{q}  \hat{f}_{E}(0)^3\\
=& \frac{-8}{q} \sum_{\alpha\in \mathrm{GF}(q)^*} \sum_{x, y, z \in E } (-1)^{\mathrm{Tr}\left ( a  \alpha x+ b \alpha y +  c \alpha z \right )}+  \frac{1}{q}  \hat{f}_{E}(0)^3\\
=& \frac{-8}{q} \sum_{x, y, z \in E }  \sum_{\alpha\in \mathrm{GF}(q)^*} (-1)^{\mathrm{Tr}\left ( \left ( a x+ b y + c z \right )\alpha \right)}+  \frac{1}{q}  \hat{f}_{E}(0)^3\\
=&\frac{-8}{q} \sum_{x, y, z \in E }  \sum_{\alpha\in \mathrm{GF}(q)} (-1)^{\mathrm{Tr}\left ( \left ( a x+ b y + c z \right )\alpha \right)}+ \frac{8}{q} |E|^3+ \frac{1}{q}  \hat{f}_{E}(0)^3\\
=& - 8 N_{E}(a, b, c)+\frac{8}{q} |E|^3+  \frac{1}{q}  \hat{f}_{E}(0)^3.
\end{align*}
The desired conclusion then follows from Part (1) of Lemma \ref{lem:one-two=0}.
\end{proof}

In order to characterize those base blocks supporting
$3$-designs, we need the next lemmas.

\begin{lemma}\label{lem:3design-equation}
Let $B$ be a $k$-subset of $\mathrm{GF}(q)$.
Let $E$ be a subset of $\mathrm{GF}(q)$ such that $\hat{f}_B(\mu)=\hat{f}_E(\mu^{d})$ for any $\mu \in \mathrm{GF}(q)$,
  where   $\mathrm{gcd}(d,q-1)=1$. Let $a, b\in \mathrm{GF}(q)^*$.
Then
\begin{align*}
\sum_{x,y\in \mathrm{GF}(q)} (-1)^{f_{B}(x)+ f_{B}(y)+f_{B}(ax+by)  }= q^2-6qk+12k^2- 8 N_E(a^{d}, b^{d}, 1).
\end{align*}
  In particular, $N_B(a, b, 1)=N_E(a^{d}, b^{d}, 1)$.
\end{lemma}
\begin{proof}
Let $S=\sum_{x,y\in \mathrm{GF}(q)} (-1)^{f_{B}(x)+ f_{B}(y)+f_{B}(ax+by)  }$.
Using Lemma \ref{lem:add-mult}, one obtains
\begin{align*}
S=\frac{1}{q} \sum_{\alpha\in \mathrm{GF}(q)} \hat{f}_{B}(a \alpha) \hat{f}_{B}(b \alpha) \hat{f}_{B}( \alpha).
\end{align*}
 From $\hat{f}_B(\mu)=\hat{f}_E(\mu^{d})$, one has
\begin{align*}
S=& \frac{1}{q} \sum_{\alpha\in \mathrm{GF}(q)} \hat{f}_{E}((a \alpha)^d) \hat{f}_{E}((b \alpha)^d) \hat{f}_{E}( \alpha^d)\\
=& \frac{1}{q} \sum_{\alpha\in \mathrm{GF}(q)} \hat{f}_{E}(a^d \alpha) \hat{f}_{E}(b^d \alpha) \hat{f}_{E}( \alpha).
\end{align*}
From Lemma \ref{lem:prodof3spec},  one has
\[S=q^2-6q|E|+12|E|^2- N_E(a^d,b^d,1).\]
From Part (1) of Lemma \ref{lem:one-two=0} and $\hat{f}_E(0)=\hat{f}_B(0)$, one obtains   $|E|=|B|=k$.
The desired conclusion then follows.
\end{proof}

\begin{lemma}\label{lem:I-N}
Let $B$ be a $k$-subset of $\mathrm{GF}(q)$.
Let $E$ be a subset of $\mathrm{GF}(q)$ such that $\hat{f}_B(\mu)=\hat{f}_E(\mu^{d})$ for any $\mu \in \mathrm{GF}(q)$,
  where   $\mathrm{gcd}(d,q-1)=1$.
Let $I_B(u_1,u_2,u_3)=|\left \{(x,y)\in \mathrm{GF}(q)^2: u_i x+y \in B \ (i=1,2,3)\right \}|$, where $u_1, u_2, u_3$ are three  pairwise distinct elements in $\mathrm{GF}(q)$.
Then
\begin{align*}
I_B(u_1,u_2,u_3)=N_{E}\left ( (u_2+u_3)^d, (u_3+u_1)^d, (u_1+u_2)^d \right ).
\end{align*}

\end{lemma}

\begin{proof}
From the definition of the function $f_{B}$, one has
\begin{align*}
8I_B(u_1,u_2,u_3)=& \sum_{x, y \in \mathrm{GF}(q)} \prod_{i=1}^3\left ( 1-(-1)^{f_{B}(u_i x+y)} \right )\\
=& q^2 - \sum_{i=1}^3 \sum_{x, y \in \mathrm{GF}(q)} (-1)^{f_{B}(u_i x+y)}\\
& + \sum_{1\le i<j \le 3} \sum_{x, y \in \mathrm{GF}(q)} (-1)^{f_{B}(u_i x+y)+f_{B}(u_j x+y)}\\
&-  \sum_{x, y \in \mathrm{GF}(q)} (-1)^{\sum_{i=1}^3 f_{B}(u_i x+y)}.
\end{align*}
By Lemma \ref{lem:one-two=0}, one gets
\begin{align*}
8I_B(u_1,u_2,u_3)=q^2-6kq+12k^2 -  \sum_{z, w \in \mathrm{GF}(q)} (-1)^{\sum_{i=1}^3 f_{B}(u_i z+w)}.
\end{align*}
With the substitution $z=\frac{x+y}{u_1+u_2}$, $w=u_1 \frac{x+y}{u_1+u_2}+x$, one obtains
\begin{align*}
8I_B(u_1,u_2,u_3)=& q^2-6kq+12k^2-  \sum_{z, w \in \mathrm{GF}(q)} (-1)^{\sum_{i=1}^3 f_{B}(u_i z+w)}\\
=& q^2-6kq+12k^2- \sum_{x, y \in \mathrm{GF}(q) } (-1)^{ f_{B}(x)+ f_{B}(y)+f_{B}\left ( \frac{(u_2+u_3)x+ (u_3+u_1)y}{u_1+u_2} \right )}.
\end{align*}
By Lemma \ref{lem:3design-equation}, one gets
\begin{align*}
I_B(u_1,u_2,u_3)=N_{E}\left ( \left (\frac{u_2+u_3}{u_1+u_2} \right )^d, \left (\frac{u_3+u_1}{u_1+u_2} \right )^d ,1\right ).
\end{align*}
Since $N_{E}(a,b,c)=N_{E}(\lambda a, \lambda b, \lambda c)$ for $a,b,c\in \mathrm{GF}(q)$ and $\lambda \in \mathrm{GF}(q)^*$, the desired conclusion then follows.
\end{proof}

The following lemma characterizes the
$3$-design from the base block $B$ by
$I_B(u_1,u_2,u_3)$.
\begin{lemma}\label{lem:lam-I}
Let $B$ be a $k$-subset of $\mathrm{GF}(q)$ with $k\ge 3$ and $\mathcal B= \mathrm{GA}_1(q)(B)$.
Define $I_B(u_1,u_2,u_3)=|\left \{(x,y)\in \mathrm{GF}(q)^2: u_i x+y \in B \  (i=1,2,3)\right \}|$,
where $u_1, u_2, u_3$ are three  pairwise distinct elements in $\mathrm{GF}(q)$.
Then, $(\mathrm{GF}(q), \mathcal B)$ is a $3$-design, if and only if, $I_B(u_1,u_2,u_3)$
is independent of the specific choice of $u_1, u_2 $ and $u_2$.
\end{lemma}

\begin{proof}
Let $\mathrm{GA}_1(q)_B=s$. Then, $\mathcal B=\{\pi_1(B), \ldots, \pi_{b} (B)\}$,
where $b=\frac{q(q-1)}{s}$ and $\pi_i\in \mathrm{GA}_1(q)$ ($1\le i \le b$),
and, $\mathrm{GA}_1(q)=\{\pi_i \pi: 1\le i \le b, \pi \in  \mathrm{GA}_1(q)_B \}$.

For any $3$-subset $\{u_1,u_2,u_3\}$ of $\mathrm{GF}(q)$, define
\[\lambda_{\{u_1,u_2,u_3\}}=|\overline{M}(u_1,u_2,u_3)|,\]
where $\overline{M}(u_1,u_2,u_3)=\left \{\pi_i(B) :1\le i \le b , \{u_1,u_2,u_3\} \subseteq \pi_i(B)\right \}$.
Note that
\[|\overline{M}(u_1,u_2,u_3)|= \frac{1}{s} M(u_1,u_2,u_3),\]
where $M(u_1,u_2,u_3)=\left\{\pi \in  \mathrm{GA}_1(q)  : \{u_1,u_2,u_3\} \subseteq \pi(B)  \right \}$.
Thus,
\begin{align*}
\lambda_{\{u_1,u_2,u_3\}}=&\frac{1}{s} |\left\{\pi \in  \mathrm{GA}_1(q)  : \pi(u_i) \in  B (i=1,2,3)  \right \}|\\
=& \frac{1}{s} |\left\{(x,y)\in \mathrm{GF}(q)^*\times \mathrm{GF}(q)   : (u_i x+y)\in  B (i=1,2,3)  \right \}|\\
=& \frac{1}{s} |\left\{(x,y)\in \mathrm{GF}(q)\times \mathrm{GF}(q)   : (u_i x+y)\in  B (i=1,2,3)  \right \}|\\
&-\frac{1}{s} |\left\{ y\in  \mathrm{GF}(q)   : y \in  B \right \}|.
\end{align*}
Thus,
\begin{align}\label{eq:repeat-norep}
\lambda_{\{u_1,u_2,u_3\}}=\frac{I_B(u_1,u_2,u_3)-k}{s}.
\end{align}
The desired conclusion then follows from Equation (\ref{eq:repeat-norep}).
\end{proof}

The following theorem presents
a characterization
of base blocks supporting 3-designs.
\begin{theorem}\label{thm:3designs-pracchar}
Let $B$ be a $k$-subset of $\mathrm{GF}(q)$ with $k\ge 3$ and $\mathcal B= \mathrm{GA}_1(q)(B)$.
Let $E$ be a subset of $\mathrm{GF}(q)$ such that $\hat{f}_B(\mu)=\hat{f}_E(\mu^{d})$ for any $\mu \in \mathrm{GF}(q)$,
  where   $\mathrm{gcd}(d,q-1)=1$.
Then, the following are equivalent:

(1) $(\mathrm{GF}(q), \mathcal B)$ is a $3$-design.

(2) $\sum_{x,y\in \mathrm{GF}(q)} (-1)^{f_{E}(x)+ f_{E}(y)+f_{E}(u^dx+(1+u)^dy)  }$ is independent of $u$, where $u\in \mathrm{GF}(q)\setminus \mathrm{GF}(2)$.

(3) $\sum_{\alpha\in \mathrm{GF}(q)} \hat{f}_{E}( \alpha) \hat{f}_{E}( u^d \alpha) \hat{f}_{E}( (1+u)^d  \alpha)$ is independent of $u$,
where $u\in \mathrm{GF}(q)\setminus \mathrm{GF}(2)$.

(4) $N_E(u^d, (1+u)^d, 1)$ is independent of $u$, where $u\in \mathrm{GF}(q)\setminus \mathrm{GF}(2)$.
\end{theorem}

\begin{proof}
From Lemmas \ref{lem:add-mult} and \ref{lem:prodof3spec}, one has
\begin{align*}
\lefteqn{ \sum_{x,y\in \mathrm{GF}(q)} (-1)^{f_{E}(x)+f_{E}(y)+f_{E}(u^d x+(1+u)^dy)} } \\
 &=
\frac{1}{q}\sum_{\alpha\in \mathrm{GF}(q)} \hat{f}_{E}(u^d \alpha) \hat{f}_{E}((1+u)^d \alpha) \hat{f}_{E}( \alpha)\\
&= q^2-6|E|q+12|E|^2- 8 N_E(u^d, (1+u)^d, 1).
\end{align*}
Hence, Parts (2), (3) and (4) are equivalent.

Assume that $(\mathrm{GF}(q), \mathcal B)$ is a $3$-design. By Lemma \ref{lem:lam-I},
$I_B(u_1,u_2,u_3)$
is independent of the specific choice of the $3$-subset $\{u_1,u_2,u_3\}$ in $ \mathrm{GF}(q)$.
By Lemma \ref{lem:I-N}, one has $N_{E}\left ( (u_2+u_3)^d, (u_3+u_1)^d, (u_1+u_2)^d \right )$ is also independent of the specific choice of the $3$-subset $\{u_1,u_2,u_3\}$ in $ \mathrm{GF}(q)$.
In particular, by choosing $u_1=0$, $u_2=1$ and $u_3=u\in  \mathrm{GF}(q)\setminus \mathrm{GF}(2)$,  $N_E(u^d, (1+u)^d, 1)$ is independent of $u$.

Conversely,  assume that $N_E(u^d, (1+u)^d, 1)$ is independent of $u$, where $u\in \mathrm{GF}(q)\setminus \mathrm{GF}(2)$.
By Lemma \ref{lem:I-N}, $I_B(u_1,u_2,u_3)$
is independent of the specific choice of the $3$-subset $\{u_1,u_2,u_3\}$ in $ \mathrm{GF}(q)$.
Thus, $(\mathrm{GF}(q), \mathcal B)$ is a $3$-design from Lemma \ref{lem:lam-I}.

It completes the proof.
\end{proof}

Choosing $E=B$ and $d=1$ in Theorem \ref{thm:3designs-pracchar},  we have the following results.
\begin{theorem}
Let $B$ be a $k$-subset of $\mathrm{GF}(q)$ with $k\ge 3$ and $\mathcal B= \mathrm{GA}_1(q)(B)$.
Then, the following are equivalent:

(1) $(\mathrm{GF}(q), \mathcal B)$ is a $3$-design.

(2) $\sum_{x,y\in \mathrm{GF}(q)} (-1)^{f_{B}(x)+ f_{B}(y)+f_{B}(ux+(1+u)y)  }$ is independent of $u$, where $u\in \mathrm{GF}(q)\setminus \mathrm{GF}(2)$.

(3) $\sum_{\alpha\in \mathrm{GF}(q)} \hat{f}_{B}( \alpha) \hat{f}_{B}( u \alpha) \hat{f}_{B}( (1+u)  \alpha)$ is independent of $u$,
where $u\in \mathrm{GF}(q)\setminus \mathrm{GF}(2)$.

(4) $N_B(u, 1+u, 1)$ is independent of $u$, where $u\in \mathrm{GF}(q)\setminus \mathrm{GF}(2)$.
\end{theorem}

The following theorem gives a characterization of a $3$-design $(\mathrm{GF}(q), \GA_1(q)(B))$
from a base block $B$ in terms of the number of solutions of some associated equations.

\begin{theorem}\label{thm:3designs-equations}
Let $B$ be a $k$-subset of $\mathrm{GF}(q)$ with $k\ge 3$ and $\mathcal B= \mathrm{GA}_1(q)(B)$.
Suppose that  $\hat{f}_B(\mu)=\sum_{x\in \mathrm{GF}(q)} (-1)^{\mathrm{Tr}(x^t+ \mu^d x)}$ for any $\mu \in \mathrm{GF}(q)$, where  $\mathrm{gcd}(td,q-1)=1$.
Then,  $(\mathrm{GF}(q), \mathcal B)$ is a $3$-design, if and only if, $|\{x \in \mathrm{GF}(q): \left (u^d x+(1+u)^d\right )^t+ x^t+1 =0 \}|$
is independent of $u$, where $u\in \mathrm{GF}(q)\setminus \mathrm{GF}(2)$.
\end{theorem}

\begin{proof}
Let $E=\{x\in \mathrm{GF}(q):  \mathrm{Tr}(x^t)=1 \}$. Then  $f_E(x)=\mathrm{Tr}(x^t)$ and $\hat{f}_B(\mu)=\hat{f}_E(\mu^d)$.
Let $S_u=\sum_{x,y\in \mathrm{GF}(q)} (-1)^{f_{E}(x)+ f_{E}(y)+f_{E}(u^dx+(1+u)^dy)  }$, where $u\in \mathrm{GF}(q)\setminus \mathrm{GF}(2)$.
Then,
\begin{align*}
S_u=& \sum_{x,y\in \mathrm{GF}(q)} (-1)^{\mathrm{Tr}\left ( x^t+y^t+\left (u^d x+(1+u)^d y \right )^t \right )}\\
=& \sum_{y\in \mathrm{GF}(q)^* } \sum_{x\in \mathrm{GF}(q)} (-1)^{\mathrm{Tr}\left ( x^t+y^t+\left (u^d x+(1+u)^d y \right )^t \right )} +
\sum_{x\in \mathrm{GF}(q)} (-1)^{\mathrm{Tr}\left ( (1+u^{dt})x^t \right )}\\
=& \sum_{y\in \mathrm{GF}(q)^* } \sum_{x\in \mathrm{GF}(q)} (-1)^{\mathrm{Tr}\left ( y^t \left ( (xy^{-1})^t+1+\left (u^d xy^{-1}+(1+u)^d  \right )^t \right ) \right )}\\
=& \sum_{y\in \mathrm{GF}(q)^* } \sum_{x\in \mathrm{GF}(q)} (-1)^{\mathrm{Tr}\left ( y^t \left ( x^t+1+\left (u^d x+(1+u)^d  \right )^t \right ) \right )}\\
=& \sum_{x\in \mathrm{GF}(q)} \sum_{y\in \mathrm{GF}(q)^* }  (-1)^{\mathrm{Tr}\left ( y \left ( x^t+1+\left (u^d x+(1+u)^d  \right )^t \right ) \right )}\\
=& \sum_{x\in \mathrm{GF}(q)} \sum_{y\in \mathrm{GF}(q)}  (-1)^{\mathrm{Tr}\left ( y \left ( x^t+1+\left (u^d x+(1+u)^d  \right )^t \right ) \right )}-q\\
=& q |\{x \in \mathrm{GF}(q): x^t+1+\left (u^d x+(1+u)^d\right )^t =0 \}|-q.
\end{align*}
The desired conclusion  then follows  from Theorem \ref{thm:3designs-pracchar}.
\end{proof}

\section{The stabilizer of the base block }\label{sec:stabi}

For any function $h$ from $\mathrm{GF}(q)$ to the finite field $\mathrm{GF}(2)$ or the field of real numbers, let $\mathrm{Supp}(h)$
denote the set $\{x\in \mathrm{GF}(q): h(x)\neq 0\}$. We have the following theorem for  the stabilizer of a base block.

\begin{theorem}\label{thm:intersection-blocks}
Let $B, E$ be two subsets of $\mathrm{GF}(q)$ such that $f_E$ is a semi-bent function from $\mathrm{GF}(q)$ to $\mathrm{GF}(2)$.
Suppose that  $\hat{f}_B(\mu)=\hat{f}_E(\mu^d)$ for any $\mu \in \mathrm{GF}(q)$
and $\mathrm{Supp}(\hat{f}_E)\neq b\cdot \mathrm{Supp}(\hat{f}_E)$
for any $b\in \mathrm{GF}(q)\setminus \mathrm{GF}(2)$,  where  $\mathrm{gcd}(d,q-1)=1$.
 Then
$$
\mathrm{GA}_1(q)_B=\{x\}.
$$
Hence, $|\mathrm{GA}_1(q)(B)|=q(q-1)$.
\end{theorem}

\begin{proof}
 Let $(b, c) \in \gf(q)^* \times \gf(q)$ such that $f_B(bx+c)=f_B(x)$. The desired conclusion is the same
as that $(b,c)=(1,0)$.

Let $S=\sum_{x\in \mathrm{GF}(q)} (-1)^{f_{B}(x)+f_{B}(bx+c)}$. Since $f_{B}(bx+c)=f_{B}(x)$, we have $S=q$.
We now compute $A$ in a different way. Note that
\begin{eqnarray*}
\sum_{\mu  \in \mathrm{GF}(q)} (-1)^{\mathrm{Tr}(\mu  (x+y))}
=\left\{
\begin{array}{ll}
q & \mbox{ if }  x=y, \\
0 & \mbox{ if }   x \neq y.
\end{array}
\right.
\end{eqnarray*}
We have then
\begin{align*}
q S=& \sum_{x,y\in \mathrm{GF}(q)} (-1)^{f_{B}(x)+f_{B}(by+c)} \sum_{\mu  \in \mathrm{GF}(q)} (-1)^{\mathrm{Tr}(\mu  (x+y))}\\
=& \sum_{\mu  \in \mathrm{GF}(q)} \sum_{x\in \mathrm{GF}(q)} (-1)^{f_{B}(x)+\mathrm{Tr}(\mu  x)} \sum_{y\in \mathrm{GF}(q)} (-1)^{f_{B}(by+c)+\mathrm{Tr}(\mu  y)}\\
=& \sum_{\mu  \in \mathrm{GF}(q)} \hat{f}_{B}(\mu ) \sum_{y\in \mathrm{GF}(q)} (-1)^{f_{B}(by+c)+\mathrm{Tr}\left (\frac{\mu }{b} (by+c)+\frac{c\mu}{b} \right )}\\
=& \sum_{\mu  \in \mathrm{GF}(q)} \hat{f}_{B}(\mu ) \hat{f}_{B}(\frac{\mu }{b}) (-1)^{\mathrm{Tr}\left (\frac{c\mu }{b}\right )}.
\end{align*}
Since $S=q$, from $\hat{f}_B(\mu)=\hat{f}_E(\mu^d)$, we then deduce that
\begin{eqnarray}\label{eq:semi-bent-blocks}
q^2=\sum_{\beta \in \mathrm{GF}(q)} \hat{f}_{E}\left (\mu^{d} \right ) \hat{f}_{E}\left ( \left (\frac{\mu}{b} \right )^{d}\right )
 (-1)^{\mathrm{Tr}\left (\frac{c\mu }{b}\right )}.
\end{eqnarray}
Using this equation,   we will prove that $(b,c)=(1,0)$.

Let $q=2^n$. Since $f_E$ is semi-bent, then $|\mathrm{Supp}(\hat{f}_E)|=2^{n-1}$ and
\begin{eqnarray*}
q^2=\sum_{\mu  \in \mathrm{GF}(q)} \hat{f}_{E}\left (\mu ^{d} \right ) \hat{f}_{E}\left ( \left (\frac{\mu }{b} \right )^{d}\right )
 (-1)^{\mathrm{Tr}\left (\frac{c\mu }{b}\right )}\le&   \sum_{\mu ^{d} \in T} 2^{\frac{n+1}{2}} 2^{\frac{n+1}{2}}\\
\le & |T| 2^{n+1}\\
\le & 2^{n-1} 2^{n+1},
\end{eqnarray*}
where $T= \mathrm{Supp}(\hat{f}_E) \cap  b^d \mathrm{Supp}(\hat{f}_{E})$.
Thus, $|T|=2^{n-1}$.
Hence, $\mathrm{Supp}(\hat{f}_E) =  b^d \mathrm{Supp}(\hat{f}_{E})$.
From the assumption of this lemma, we must have $b=1$.

Since $b=1$,   Equation (\ref{eq:semi-bent-blocks})  becomes
\begin{eqnarray}
q^2=2^{n+1} \sum_{ \mu^{d}\in  \mathrm{Supp}(\hat{f}_E)}  (-1)^{\mathrm{Tr}\left (c\mu \right )}.\nonumber
\end{eqnarray}
This equation forces $\tr(c\mu )=0$ for all the $2^{n-1}$ nonzero elements $\mu  \in \gf(q)$
such that $ \mu ^{d}\in  \mathrm{Supp}(\hat{f}_E)$. Note that $\tr(c \times 0)=0$. Thus, $\tr(cx)=0$ has at least $2^{n-1}+1$
solutions, which holds only if $c=0$. This completes the proof.
\end{proof}

\section{Two constructions   of $3$-designs from  APN functions}\label{sec-construct}
The objective of this section
 is to construct  $3$-designs from APN functions over $\mathrm{GF}(q)$, where $q=2^n$.

Recall that a function $F$ from $\mathrm{GF}(q)$ to itself is called \emph{almost perfect nonlinear (APN)},
if $F(x+a)+F(x)=b$ has at most two solutions in $\mathrm{GF}(q)$ for every pair $(a,b)\in \mathrm{GF}(q)^*\times \mathrm{GF}(q)$.
$F$ is said to be \emph{almost bent (AB)} if $\mathcal W_{F}(a,b)=0$, or $\pm 2^{\frac{n+1}{2}}$
 for every pair $(a,b)$ with $a\neq 0$, where
 \begin{align*}
 \mathcal W_{F}(a,b)=\sum_{x\in \mathrm{GF}(q)} (-1)^{\mathrm{Tr}(aF(x)+bx)}.
 \end{align*}

The following is a list of known APN power functions
over $\mathrm{GF}(2^n)$ of the form $F(x)=x^s$:
\begin{itemize}
\item $s=2^i+1$ with $\mathrm{gcd}(i,n)=1$ (Gold functions);
\item $s=2^{2i}-2^i+1$ with $\mathrm{gcd}(i,n)=1$ (Kasami functions);
\item $s=2^{\frac{n-1}{2}}+3$ with $n$ odd ( Welch functions);
\item $s=2^{\frac{n-1}{2}}+2^{\frac{n-1}{4}}-1$ with $n\equiv 1 \pmod{4}$,\\
 $s=2^{\frac{n-1}{2}}+2^{\frac{3n-1}{4}}-1$ with $n\equiv 3 \pmod{4}$ (Niho functions);
\item $s=2^n-2$ with $n$ odd (inverse functions);
\item $s=2^{\frac{4n}{5}}+ 2^{\frac{3n}{5}}+ 2^{\frac{2n}{5}}+ 2^{\frac{n}{5}}-1$ with $n\equiv 0 \pmod{5}$  (Dobbertin functions).
\end{itemize}
When $n$ is odd, Gold functions, Kassami functions, Welch functions and Niho functions over $\mathrm{GF}(2^n)$ are AB functions.
We will present two constructions of
$3$-designs from APN functions.
 \subsection{The first construction of $3$-designs from Kassami APN functions}\label{subsec:des-Kassami}
 Let $q=2^n$ with $n$ odd. We follow the convention that if $e$  is an exponent of a power function over $\gf{(q)}$, then $\frac{1}{e}$  is interpreted as the inverse of $e$  modulo $(q-1)$. Thus $\frac{1}{e}$ exists if and only if $\mathrm{gcd}(e,q-1)=1$.
 Let $x^{2^{2i}-2^i+1}$ be the Kassmi power function over $\gf{(q)}$ with $\mathrm{gcd}(i,n)=1$.
 Since $\mathrm{gcd}(2,n)=1$ and $\mathrm{gcd}(i,n)=1$, one has  $\mathrm{gcd}(3,q-1)=\mathrm{gcd}(2^i+1,q-1)=1$.
 Thus $\frac{1}{3}$ and $\frac{1}{2^i+1}$ exist.
 Define
 \begin{align}\label{eq:B-K}
 B=\mathrm{GF}(q)\setminus \left \{ \left ((x+1)^s+x^s+1 \right )^{\frac{1}{2^i+1}}: x\in \mathrm{GF}(q) \right \},
 \end{align}
where $s=2^{2i}-2^i+1$. In this case, we also denote the base block $B$ by $KA_{n,i}$.
We shall study the incidence  structure
\[\mathbb{KA}_{n,i}=\left (\gf{(2^n)},  \mathrm{GA}_1(2^n)(KA_{n,i}) \right ).\]

 The Walsh coefficients of $f_B$ is given in the following lemma \cite[Lemma A1]{DD04}.

 \begin{lemma}\label{lem:Dil-Dob}
 Let  $n$ be an odd integer and let $i$ be a positive integer with $\mathrm{gcd}(i,n)=1$. Let $B$ be the subset of $\mathrm{GF}(q)$ given by (\ref{eq:B-K}).
 Then for all $\mu \in \gf{(q)}$ we have
\[\hat{f}_B(\mu)=\hat{f}_E\left (\mu^{\frac{2^i+1}{3}} \right )=\sum_{x \in \gf{(q)}} (-1)^{\mathrm{Tr}\left (x^3+\mu^{\frac{2^i+1}{3}} x\right )},\]
where $E=\{x \in \gf{(q)}: \mathrm{Tr}(x^3)=1\}$.
 \end{lemma}

 A major result of this paper is the following.

 \begin{theorem}\label{thm:Kassami}
 Let  $n$ be an odd integer and let $i$ be a positive integer with $\mathrm{gcd}(i,n)=1$. Let $B$ be the base block  given by (\ref{eq:B-K}).
 Then the incidence structure $\mathbb{KA}_{n,i}=(\mathrm{GF}(q), \mathrm{GA}_1(q)(B))$ is a $3$-$\left (q, \frac{q}{2}, \frac{q(q-4)}{8} \right )$ design.
 \end{theorem}
It is observed that $\mathbb{KA}_{n,i}$ and $\mathbb{KA}_{n,n-i}$ are isomorphic. Thus, we only need to consider
the $3$-design $\mathbb{KA}_{n,n-i}$, where $1\le i \le \frac{n-1}{2}$ and $\mathrm{gcd}(i,n)=1$.

To prove Theorem \ref{thm:Kassami}, we need the following lemmas.
\begin{lemma}\label{lem:solom}
Let $\sigma_1, \sigma_2, \sigma_3 \in \gf{(q)}$ such that  $\sigma_1^2\neq \sigma_2$ and $\sigma_3\neq \sigma_1 \sigma_2$. Then
the cubic equation $x^3+\sigma_1 x^2 + \sigma_2 x + \sigma_3 =0$ has a unique solution
$x\in \mathrm{GF}(2^n)$, if and only if
\[\mathrm{Tr}\left( \frac{(\sigma_2+\sigma_1^2)^3}{(\sigma_3+\sigma_1 \sigma_2)^2}+1\right )=1.\]
\end{lemma}
\begin{proof}
The desired conclusion  follows from \cite[Theorem 2]{BRS67}.
\end{proof}

 \begin{lemma}\label{lem:nonsymmtry}
 Let  $n$ be an odd integer and $i$ be a positive integer with $\mathrm{gcd}(i,n)=1$. Then the cubic equation $\left (u^d x+(1+u)^d\right )^3+ x^3+1 =0 $
 has a unique solution $x\in \mathrm{GF}(2^n)$, where $d \equiv \frac{2^i+1}{3} \pmod {2^n-1}$ and  $u \in \gf{(q)} \setminus \gf{(2)}$.
 \end{lemma}
\begin{proof}
Let $S=\left (u^d x+(1+u)^d\right )^3+ x^3+1$.
Expanding $\left (u^d x+(1+u)^d\right )^3$, one has
\begin{align*}
S=& \left ( 1+ u^{3d} \right ) x^3 +  u^{2d} (1+u)^d  x^2 +   u^{d} (1+u)^{2d} x+ \left ( (1+u)^{3d} +1  \right ).
\end{align*}
Note that $\mathrm{gcd}(3d,2^n-1)=1$ and $u\neq 1$. Then $1+ u^{3d}\neq 0$ and
\begin{align*}
\frac{S}{1+ u^{3d}}=& x^3 +  \frac{u^{2d} (1+u)^d}{1+ u^{3d}}  x^2 + \frac{  u^{d} (1+u)^{2d}}{1+ u^{3d}} x+ \frac{(1+u)^{3d} +1  }{1+ u^{3d}}\\
=& x^3+\sigma_1 x^2 + \sigma_2 x + \sigma_3.
\end{align*}
One has
\begin{align}\label{eq:sgm2-sgm1}
\sigma_2+\sigma_1^2=& \frac{u^{d} (1+u)^{2d} \left ( 1+ u^{3d}\right )+  u^{4d} (1+u)^{2d}}{\left ( 1+ u^{3d}\right )^2}\nonumber \\
=& \frac{u^{d} (1+u)^{2d} \left ( \left ( 1+ u^{3d}\right )+  u^{3d}  \right )}{\left ( 1+ u^{3d}\right )^2}\nonumber\\
=& \frac{u^{d} (1+u)^{2d}}{\left ( 1+ u^{3d}\right )^2}
\end{align}
and
\begin{align}\label{eq:sgm3-sgm1sgm2}
\sigma_3+\sigma_1 \sigma_2= & \frac{\left( (1+u)^{3d} +1  \right ) \left ( 1+ u^{3d}\right )+ u^{3d} (1+u)^{3d}}{\left ( 1+ u^{3d}\right )^2}\nonumber\\
=&  \frac{ u^{3d} + (1+u)^{3d}+1}{\left ( 1+ u^{3d}\right )^2}\nonumber\\
=& \frac{ u^{2^i+1} + (1+u)^{2^i+1}+1}{\left ( 1+ u^{3d}\right )^2}\nonumber\\
=& \frac{ u^{2^i} + u}{\left ( 1+ u^{3d}\right )^2}.
\end{align}
By Equations (\ref{eq:sgm2-sgm1}) and (\ref{eq:sgm3-sgm1sgm2}) , $\sigma_2+\sigma_1^2\neq 0$ and $\sigma_3+\sigma_1 \sigma_2\neq 0$
since $u\in \gf{(2^n)}\setminus \gf{(2)}$ and $\mathrm{gcd}(i,n)=1$.

Let $A=\frac{(\sigma_2+\sigma_1^2)^3}{(\sigma_3+\sigma_1 \sigma_2)^2}$. Using Equations (\ref{eq:sgm2-sgm1}) and (\ref{eq:sgm3-sgm1sgm2}) again,
one has
\begin{align*}
A=& \frac{u^{3d} (1+u)^{6d} \left ( 1+ u^{3d}\right )^4}{\left ( 1+ u^{3d}\right )^6 \left(  u^{2^i} + u \right)^2}\\
=& \frac{u^{3d} (1+u)^{6d}}{\left ( 1+ u^{3d}\right )^2 \left(  u^{2^i} + u \right)^2}\\
=& \frac{u^{2^i+1} (1+u)^{2^{i+1}+2}}{\left ( 1+ u^{2^i+1}\right )^2 \left(  u^{2^i} + u \right)^2}\\
=& \frac{UV(U+1)^2(V+1)^2}{(U+V)^2(UV+1)^2},
\end{align*}
where $U=u$ and $V=u^{2^i}$.
It is observed that
\begin{align*}
\frac{UV(U+1)^2(V+1)^2}{(U+V)^2(UV+1)^2}=W^2+W,
\end{align*}
where $W=\frac{U(V+1)^2}{(U+V)(UV+1)}$.
Thus, $\mathrm{Tr}\left( \frac{(\sigma_2+\sigma_1^2)^3}{(\sigma_3+\sigma_1 \sigma_2)^2}+1\right )=\mathrm{Tr}(W^2+W+1)=1$.
The desired conclusion then follows from Lemma \ref{lem:solom}.
\end{proof}

\begin{proof}[\textbf{The proof of Theorem \ref{thm:Kassami}}]

Using Lemma  \ref{lem:Dil-Dob}, one has, for all $\mu \in \gf{(q)}$,
\[\hat{f}_B(\mu)=\hat{f}_E\left (\mu^{\frac{2^i+1}{3}} \right )=\sum_{x \in \gf{(q)}} (-1)^{\mathrm{Tr}\left (x^3+\mu^{\frac{2^i+1}{3}} x\right )},\]
where $E=\{x \in \gf{(q)}: \mathrm{Tr}(x^3)=1\}$.
Combining  Lemma \ref{lem:nonsymmtry} and Theorem \ref{thm:3designs-equations},  one has $(\mathrm{GF}(q), \mathrm{GA}_1(q)(B))$ is a $3$-design.

Note that $f_E(x)=\mathrm{Tr}\left (x^3\right)$ is a semi-bent function \cite{Sihem16}, and $\mathrm{Supp}(\hat{f}_E)=\{\mu \in \mathrm{GF}(q): \mathrm{Tr}(\mu)=1\}$, which is  a Singer difference set with parameters $(q-1, q/2, q/4)$ in the group $\mathrm{GF}(q)^*$.
Thus, $|b\cdot \mathrm{Supp}(\hat{f}_E) \cap \mathrm{Supp}(\hat{f}_E)|=q/4$ for any $b\in \mathrm{GF}(q)\setminus \mathrm{GF}(2)$ \cite[Proposition 2]{DD04}.
Then, $|\mathrm{GA}_1(q)(B)|=q(q-1)$ from Theorem \ref{thm:intersection-blocks}.
From Part (1) of  Lemma \ref{lem:one-two=0} and $\hat{f}_B(0)= \hat{f}_E(0) $, one has $|B|=|E|=q/2$.

Hence,  the incidence structure $(\mathrm{GF}(q), \mathrm{GA}_1(q)(B))$ is a $3$-$\left (q, \frac{q}{2}, \frac{q(q-4)}{8} \right )$ design.
This completes the poof of Theorem \ref{thm:Kassami}.

\end{proof}

 \subsection{Another construction of $3$-designs from APN functions}\label{subsec:des-apn}

 Let $x^s$ be an APN function over $\mathrm{GF}(q)$ with $\mathrm{gcd}(s,q-1)=1$.
Define the base block $B_s$  as
\begin{align}\label{eq:Bs}
B_{s}=\{(x+1)^s+x^s: x\in \mathrm{GF}(q)\}.
\end{align}
Since $x^s$ is APN, the function
$(x+1)^s+x^s$ is $2$-to-$1$. Thus, $|B_s|=\frac{q}{2}$.
In this case, we also denote the base block $B_s$ by $AP_{n,s}$.
We shall study the incidence  structure
\[\mathbb{AP}_{n,s}=\left (\gf{(2^n)},  \mathrm{GA}_1(2^n)(AP_{n,s}) \right ).\]

When $s=2^i+1$, we have the following theorem on $3$-designs
$\mathbb{AP}_{n,s}$.
\begin{theorem}
Let $n \geq 4$ and $s=2^i+1$, where  $n/\gcd(i,n)$ is odd. Then the incidence structure
$\mathbb{AP}_{n,s}=(\gf(q), \mathrm{GA}_1(q)(B_s))$ is a $3$-$(q, q/2, (q-4)/4)$.
\end{theorem}

\begin{proof}
Note that $(x+1)^s-x^s$ is an affine function. The proof is similar to the proof
of Corollary 29 in \cite{DT192}.
\end{proof}

When $s=2^{2i}-2^i+1$, we need the following lemma to characterize $3$-designs
$\mathbb{AP}_{n,s}$.
\begin{proposition}\label{prop:val-spec}
Let $n=3 i \pm 1$ and $s=2^{2i}-2^i+1$, where $i$ is an even positive integer.
Then,
$$\hat{f}_{B_s}(\mu)= \hat{f}_{E}(\mu^{d})=\sum_{x\in \mathrm{GF}(q)} (-1)^{\mathrm{Tr}\left ( x^{2^i+1} +\mu^d x \right )},$$
where $E=\{\mu: \mathrm{Tr}(x^{2^i+1})=1\}$ and $d\equiv \frac{1}{s} \pmod{2^n-1}$.
\end{proposition}
\begin{proof}
If $\mu=0$, then from Part (1) of Lemma \ref{lem:one-two=0} one has $\hat{f}_{B_s}(\mu)= \hat{f}_{E}(\mu^{d})$.

Suppose $\mu\in \mathrm{GF}(q)^*$. Then $\sum_{x \in \mathrm{GF}(q)} (-1)^{ \mathrm{Tr}(\mu x)}=0$ and
\begin{align*}
\hat{f}_{B_s}(\mu)=& \sum_{x \in \mathrm{GF}(q)} (-1)^{f_{B_s}(x)+ \mathrm{Tr}(\mu x)}\\
=& \sum_{x \in \mathrm{GF}(q)} \left ((-1)^{f_{B_s}(x)}-1 \right )(-1)^{ \mathrm{Tr}(\mu x)}\\
=& \sum_{x \in B_s}  (-2)\cdot  (-1)^{ \mathrm{Tr}(\mu x)}.
\end{align*}
Since $(x+1)^s+x^s$ is a two-to-one function from $\mathrm{GF}(q)$ to $B_s$, one has
\begin{align*}
\hat{f}_{B_s}(\mu)=& \sum_{x \in \mathrm{GF}(q)}  (-1)\cdot  (-1)^{ \mathrm{Tr}(\mu  ((x+1)^s+x^s))}\\
=& \frac{-1}{q} \sum_{x, y \in \mathrm{GF}(q)}   (-1)^{ \mathrm{Tr}(\mu (x^s+y^s))} \sum_{\beta  \in \mathrm{GF}(q) } (-1)^{\mathrm{Tr}(\beta (x+y+1))}\\
=& \frac{-1}{q} \sum_{\beta  \in \mathrm{GF}(q) } (-1)^{\mathrm{Tr(\beta)}} \left ( \sum_{x \in \mathrm{GF}(q)} (-1)^{ \mathrm{Tr}(\mu  x^s+\beta x)} \right )^2.
\end{align*}
Note that $\mathrm{gcd}(s, 2^n-1)=1$.
Then $\sum_{x \in \mathrm{GF}(q)} (-1)^{ \mathrm{Tr}(\mu  x^s+\beta x)}= \sum_{x \in \mathrm{GF}(q)} (-1)^{ \mathrm{Tr}(x^s+\mu^{-\frac{1}{s}}\beta x)} $, and
\begin{align}
\hat{f}_{B_s}(\mu)=&  \frac{-1}{q} \sum_{\beta  \in \mathrm{GF}(q) } (-1)^{\mathrm{Tr}(\beta)}  (\hat{g}(\mu^{-\frac{1}{e}}\beta))^2,\nonumber \\
=&   \frac{-1}{q} \sum_{\beta  \in \mathrm{GF}(q) } (-1)^{\mathrm{Tr}( \mu^{\frac{1}{s}} \beta)}  (\hat{g}(\beta))^2,
\end{align}
where $g(x)=\mathrm{Tr}(x^s)$.
From \cite{Dil99}, one has
\begin{align*}
\hat{g}(\beta)=
\left\{
  \begin{array}{ll}
    \pm 2^{\frac{n+1}{2}}, & \beta \in E,\\
 0, & \text{otherwise},
  \end{array}
\right.
\end{align*}
 where $E=\{\mu \in  \mathrm{GF}(2^n):  \mathrm{Tr}(\mu^{2^i+1})=1\}$. Then
\begin{align}\label{eq:J-E}
\hat{f}_{B_s}(\mu)=&   -2  \sum_{x  \in E } (-1)^{\mathrm{Tr}( \mu^{\frac{1}{s}} x)}.
\end{align}
Using Equation (\ref{eq:J-E}), one gets
\begin{align*}
\hat{f}_{B_s}(\mu)=& \sum_{x\in \mathrm{GF}(2^n)} \left ((-1)^{f_E(x)}-1 \right )(-1)^{\mathrm{Tr}( \mu^{\frac{1}{s}} x)}\\
=& \sum_{x\in \mathrm{GF}(2^n)} (-1)^{\mathrm{Tr}( f_E(x)+\mu^{\frac{1}{s}} x)}.
\end{align*}
The desired conclusion then follows.
\end{proof}

We have the  following proposition on the incidence structure
$\mathbb{AP}_{n,s}$, where $s=2^{2i}-2^i+1$.
\begin{proposition}\label{APns-equation}
Let $n=3 i \pm 1$ and $s=2^{2i}-2^i+1$, where $i$ is an even positive integer. Let $d\equiv \frac{1}{s} \pmod{2^n-1}$.
Then, the incidence structure $\mathbb{AP}_{n,s}=(\mathrm{GF}(q), \mathrm{GA}_1(q)(B_s))$ is a $3$-$\left (q, \frac{q}{2}, \frac{q(q-4)}{8} \right )$ design, if and only if,
$$\left |\left \{x\in \mathrm{GF}(2^n): \left (u^d x+(1+u)^d\right )^{2^i+1}+ x^{2^i+1}+1 =0\right  \} \right |$$
is independent of $u$, where $u\in \mathrm{GF}(q)\setminus \mathrm{GF}(2)$.
\end{proposition}

\begin{proof}
By Proposition \ref{prop:val-spec}, one has
$$\hat{f}_{B_s}(\mu)= \hat{f}_{E}(\mu^{d})=\sum_{x\in \mathrm{GF}(q)} (-1)^{\mathrm{Tr}\left ( x^{2^i+1} +\mu^d x \right )},$$
where $E=\{\mu: \mathrm{Tr}(x^{2^i+1})=1\}$.
Note that $f_E(x)=\mathrm{Tr}\left (x^{2^i+1}\right)$ is a semi-bent function \cite{Sihem16}, and
$\mathrm{Supp}(\hat{f}_E)=\{\mu \in \mathrm{GF}(q): \mathrm{Tr}(\mu)=1\}$
, which is  a Singer difference set with parameters $(q-1, q/2, q/4)$ in the group $\mathrm{GF}(q)^*$.
Thus, for any $b\in \mathrm{GF}(q)\setminus \gf{(2)}$, $|b\cdot \mathrm{Supp}(\hat{f}_E) \cap \mathrm{Supp}(\hat{f}_E)|=q/4$.
By Theorem \ref{thm:intersection-blocks}, $|\mathrm{GA}_1(q)(B_s)|=q(q-1)$.
The desired conclusions follow from Theorem \ref{thm:3designs-equations}.
\end{proof}

In accordance with Proposition \ref{APns-equation}, we propose the following conjecture,
which is useful for obtaining $3$-designs.

\begin{conj}\label{conj:3x+1}
Let $n=3 i \pm 1$ and $s=2^{2i}-2^i+1$, where $i$ is an even positive integer.  Let $u\in \mathrm{GF}(q)\setminus \gf{(2)}$.
Then, the equation
\[\left (u^d x+(1+u)^d\right )^{2^i+1}+ x^{2^i+1}+1 =0\]
has a unique solution $x\in \mathrm{GF}(2^n)$, where  $d\equiv \frac{1}{s} \pmod{2^n-1}$.
\end{conj}
Conjecture \ref{conj:3x+1} was confirmed by Magma for $n \in \{5, 7, 11, 13\}$. If  Conjecture \ref{conj:3x+1} is true, the base block $B_{s}\subseteq \mathrm{GF}(2^n)$ supports  a $3$-design, where $n=6 i \pm 1$ and $s=2^{4i}-2^{2i}+1$.

The equation $\left (u^d x+(1+u)^d\right )^{2^i+1}+ x^{2^i+1}+1 =0$ may be reduced to $P_a(x)=x^{2^i+1}+x+a=0$, which has been considered in \cite{HK08,HK10,KM19}.
When $i$ is coprime to $n$, denote $i'= \frac{1}{i} \pmod{n}$. Recall the following sequences of polynomials that were introduced by Dobbertin in \cite{Dob99}:
\begin{align*}
A_1(x)=&x,\\
A_2(x)=&x^{2^i+1},\\
A_{j+2}(x)=&x^{2^{i(j+1)}} A_{j+1}(x)+x^{2^{i(j+1)}-2^{ij}}A_{j}(x), \text{ for } j\ge 1,\\
B_1(x)=&0,\\
B_2(x)=&x^{2^i-1},\\
B_{j+2}(x)=&x^{2^{i(j+1)}} A_{j+1}(x)+x^{2^{i(j+1)}-2^{ij}}A_{j}(x), \text{ for } j\ge 1.
\end{align*}
Define the polynomial $R_{n,i}(x)$ as
\begin{align*}
R_{n,i}(x)=\sum_{j=1}^{i'} A_j(x) +B_{i'}(x).
\end{align*}
For $i'=1, 2, 3$,  the  polynomials $R_{n, \frac{1}{i'}}(x)$ is given by
\begin{align*}
R_{n,1}(x)=& x,\\
R_{n,\frac{1}{2}}(x)=& x^{2^{i}+1}+x^{2^{i}-1}+x,\\
R_{n,\frac{1}{3}}(x)=& x^{2^{2i}+2^i+1}+x^{2^{2i}+2^i-1}+x^{2^{2i}-2^i+1}+x^{2^{i}+1}+x,
\end{align*}
where $\frac{1}{i'}$ is  the smallest positive integer $i$ such that $i i' \equiv 1 \pmod{n}$.
To solve  Conjecture \ref{conj:3x+1}, one may need the following results \cite{HK08}.

\begin{theorem}
For any $a\in \mathrm{GF}(2^n)^*$ and a positive integer $i<n$ with $\mathrm{gcd}(n,i)=1$, the  polynomial $P_a(x)=x^{2^i+1}+x+a$ has either none, one, or three zeros in $\mathrm{GF}(2^n)$.
Further, $P_a(x)$ has exactly one zero in $\mathrm{GF}(2^n)$ if and only if  $\mathrm{Tr}\left (R_{n,i}(a^{-1})+1 \right )=1$.
\end{theorem}

Next, we present some conjectures on $3$-designs. These conjectures have been confirmed by Magma for $n \in \{5, 7\}$.
\begin{conj}\label{conj:K}
Let $n \geq 5$ be odd. Let $s=2^{2i}-2^i+1$ with $\mathrm{gcd}(3i,n)=1$. Then
the incidence structure $\mathbb{AP}_{n,s}=(\mathrm{GF}(q), \mathrm{GA}_1(q)(B_s))$
is a $3$-$\left(q, \frac{q}{2}, \frac{q(q-4)}{8} \right )$ design.
\end{conj}
\begin{conj}\label{conj:2pow+3}
Let $n \geq 5$ be odd and $s=2^{\frac{n-1}{2}}+3$. Then
the incidence structure $\mathbb{AP}_{n,s}=(\mathrm{GF}(q), \mathrm{GA}_1(q)(B_s))$
is a $3$-$\left(q, \frac{q}{2}, \frac{q(q-4)}{8} \right )$ design.
\end{conj}
\begin{conj}\label{conj:2p+2p-1}
Let $n \geq 5$ be odd. Then the incidence structure $\mathbb{AP}_{n,s}=(\mathrm{GF}(q), \mathrm{GA}_1(q)(B_s))$ is a $3$-$\left(q, \frac{q}{2}, \frac{q(q-4)}{8} \right )$   design, where
\begin{itemize}
\item  $s=2^{\frac{n-1}{2}}+2^{\frac{n-1}{4}}-1$ with $n\equiv 1 \pmod{4}$;
 \item $s=2^{\frac{n-1}{2}}+2^{\frac{3n-1}{4}}-1$ with $n\equiv 3 \pmod{4}$ .
\end{itemize}
\end{conj}

\subsection{A comparison with other constructions}

Let $\mathbb{KA}_{n,i}$ and $\mathbb{AP}_{n,s}$ be the incidence structures constructed in Subsections \ref{subsec:des-Kassami} and \ref{subsec:des-apn},
respectively. The incidence structures  $\mathbb{KA}_{n,i}$ and $\mathbb{AP}_{n,s}$ are both constructed from APN functions, which are a special type of polynomials over a finite fields. By a special polynomial we mean a polynomial either of special form or with special property. For instance, monomials and permutation polynomials are special polynomials. Special polynomials have interesting applications in combinatorial designs.  The Dickson polynomials $x^5+ax^3+a^2x$ over $\gf(3^m)$ led to
a $70$-year breakthrough in searching for new skew Hadamard difference sets \cite{DY06}.
Recently, Ding and Tang  \cite{DT192} presented two constructions of $t$-designs with oval polynomials (o-polynomials) over $\gf(q)$ and obtained some families of $2$-designs and $3$-designs.

A polynomial $F$ from $\mathrm{GF}(2^n)$  to itself is called an \emph{o-polynomial} if
$F(x)+ax$ is $2$-to-$1$ for every $a\in \gf{(2^n)}^*$ \cite{CM11}.
If an o-polynomial $F(x)$ is monomial, $F(x)$ is also called an o-monomial.
Known o-monomials over $\mathrm{GF}(2^n)$ are listed:
\begin{itemize}
\item $\mathrm{Trans}_{n,i}(x)=x^{2^i}$, $\mathrm{gcd}(i,n)=1$.
\item $\Segre_{n}(x)=x^{6}$, $n$ odd.
\item $\Glynnone_n(x)=x^{3\times 2^{(n+1)/2}+4}$, $n$ odd.
\item $
\Glynntwo(x)= \left\{
\begin{array}{r}
x^{2^{(n+1)/2}+2^{(3n+1)/4}}   \mbox{ if } n \equiv 1 \pmod{4}, \\
x^{2^{(n+1)/2}+2^{(n+1)/4}}    \mbox{ if } n \equiv 3 \pmod{4}  .
\end{array}
\right.
$
\end{itemize}
For any permutation polynomial $F(x)$ over $\gf(q)$, we define $\overline{F}(x)=xF(x^{q-2})$.
If $F$ is an o-polynomial over $\gf(q)$, then $\overline{F}$ is also an o-polynomial.

Let $x^s$ be an o-monomial over $\gf{(2^n)}$. Let $\mathbb{OV}_{n,s}$ be the incidence structure
$( \gf{(2^n)}$, $\mathrm{GA}_1(q)(OV_{n,s}))$, where
\[OV_{n,s}=\{x^s+x: x \in \gf{(2^n)}\}. \]
It follows from Lemma 3 of \cite{DT192} and its proof that $\mathbb{OV}_{n,s}$ is the same as the incidence
structure $\mathbb{D}(x^s,2^{n-1})=\left (\mathrm{GF}(2^n), \mathcal{B}_{x^s,2^{n-1}} \right )$ introduced in \cite{DT192}.
In \cite{DT192}, the authors proved that $\mathbb{OV}_{n,s}$ is a $3$-$\left (2^n, 2^{n-1}, 2^{n-3}(2^n-4) \right )$ design
if $x^s$ is the o-monomial $\Segre_{n}(x)$ or $\Glynntwo(x)$, and conjectured that the same conclusion is still true
 if $x^s$ is any o-monomial with $x^s\neq \mathrm{Trans}_{n,i}(x)$.

 In general, it is extremely difficult to solve the isomorphy problem of $t$-designs theoretically.
 We have done an isomorphic  classification for the following set of $3$-designs in \cite{DT192} and this paper  for the case $n=5$ with Magma:
 \[\left \{\mathbb{KA}_{5,1}, \mathbb{KA}_{5,2}, \mathbb{AP}_{5,5}, \mathbb{AP}_{5,7} ,\mathbb{AP}_{5,13}, \mathbb{OV}_{5,6},
 \mathbb{OV}_{5,26},\mathbb{OV}_{5,28},\mathbb{OV}_{5,4},\mathbb{OV}_{5,24},\mathbb{OV}_{5,8} \right \}. \]

 Designs $\mathbb{AP}_{5,5}$, $\mathbb{AP}_{5,7}$ and $\mathbb{AP}_{5,13}$ are the
 $3$-designs from Niho APN function, Welch APN function and Kassami APN function, respectively.
 Designs $\mathbb{OV}_{5,6}$,
 $\mathbb{OV}_{5,26}$,$\mathbb{OV}_{5,28}$,
 $\mathbb{OV}_{5,4}$,
 $\mathbb{OV}_{5,24}$ and $\mathbb{OV}_{5,8}$
 are $3$-designs from o-monomials $\Segre_{5}(x)$, $\overline{\Segre}_{5}(x)$,  $\Glynnone_5(x)$, $\overline {\Glynnone}_5(x)$, $\Glynntwo_5(x)$
 and $\overline {\Glynntwo}_5(x)$, respectively.
 These $3$-designs are divided into seven distinct equivalence classes:
 $\{\mathbb{KA}_{5,1}\}$,
 $\{\mathbb{KA}_{5,2}\}$,
 $\{\mathbb{AP}_{5,7}\}$,
$\{\mathbb{OV}_{5,24}\}$,
  $\{\mathbb{OV}_{5,28}\}$,
 $\{\mathbb{AP}_{5,5}$, $\mathbb{OV}_{5,4}$, $\mathbb{OV}_{5,8}\}$,
 $\{\mathbb{AP}_{5,13}$, $\mathbb{OV}_{5,6}$, $\mathbb{OV}_{5,26}\}$.
Based on the above discussion, we propose the following conjecture, which have been confirmed by Magma for $n \in \{5, 7\}$.
\begin{conj}\label{conj:phi/2}
Let $n\ge 5$ be odd and $i\in \left \{i: 1\le i \le \frac{n-1}{2}, \mathrm{gcd}(i,n)=1\right \}$. Let $\phi(n)$
denote the Euler's totient function.
Then the $\frac{\phi(n)}{2}$ $3$-designs $\mathbb{KA}_{n,i}$ are pairwise non-isomorphic. Further, they
are  not equivalent to  any designs presented in Subsection \ref{subsec:des-apn} from APN power functions,
and are slso not equivalent to any designs introduced in \cite{DT192} from o-monomials.
\end{conj}

\subsection{Linear codes from the $3$-designs of this paper}

A $t$-design $\mathbb  D=(\mathcal P, \mathcal B)$ induces a linear code over GF($p$) for
any prime $p$.
Let $\mathcal P=\{p_1, \dots, p_{\nu}\}$. For any block $B\in \mathcal B$,   the \emph{characteristic vector} of $B$ is defined
by the vector $\mathbf{c}_{B}
=(c_1, \dots, c_{\nu})\in \{0,1\}^{\nu}$, where
\begin{align*}
c_i=
\left\{
  \begin{array}{ll}
    1, & \text{if}~ p_i \in B, \\
    0, & \text{if}~ p_i \not \in B.
  \end{array}
\right.
\end{align*}
For a prime $p$, a \emph{linear code} $\mathsf{C}_{p}(\mathbb D)$ over the prime field $\mathrm{GF}(p)$ from the design $\mathbb D$ is spanned by the characteristic vectors of the blocks of $\mathbb B$, which is
 the subspace $\mathrm{Span}\{\mathbf{v}_{B}: B\in \mathcal B\}$ of the vector space $\mathrm{GF}(p)^{\nu}$.
Linear codes $\mathsf{C}_{p}(\mathbb D)$ from designs $\mathbb D$ have been studied and documented in the literature
\cite{AK92,Ding15,Ton98,Ton07}.
For the codes of $3$-designs in this paper, we propose the following two conjectures.
\begin{conj}\label{conj:code3-APN}
Let $n\ge 5$ be odd. Then  the binary linear code $\mathsf{C}_2(\mathbb{KA}_{n,1})$ has parameters $[2^{n}, 2n+1, 2^{n-1} -2^{\frac{n-1}{2}}]$ and weight enumerator
\[1+uz^{2^{n-1} -2^{\frac{n-1}{2}}}+vz^{2^{n-1}}+uz^{2^{n-1} +2^{\frac{n-1}{2}}}+z^{2^n},\]
where $u=2^{2n-1}-2^{n-1}$ and $v=2^{2n}+2^{n}-2$. In addition, the dual code $\mathsf{C}_2(\mathbb{KA}_{n,1})^{\perp}$ has parameters $[2^n,2^n-2n-1,6]$.
\end{conj}
\begin{conj}\label{conj:code-APN}
Let $n\ge 5$ be odd and $i\in \left \{i: 1\le i \le \frac{n-1}{2}, \mathrm{gcd}(i,n)=1\right \}$. Then
the $\frac{\phi(n)}{2}$ binary linear codes $\mathsf{C}_2(\mathbb{KA}_{n,i})$ are pairwise inequivalent.
\end{conj}
Conjecture \ref{conj:code3-APN} was confirmed by Magma for $n \in \{5, 7\}$.
If Conjecture \ref{conj:code3-APN} is true, then $\mathsf{C}_2(\mathbb{KA}_{n,1})$ holds three $3$-designs, and $\mathsf{C}_2(\mathbb{KA}_{n,1})^{\perp}$
 holds exponentially many $3$-designs (see \cite{DL17} for detail).

Conjecture \ref{conj:code-APN} was confirmed by Magma for $n \in \{5, 7\}$. If Conjecture \ref{conj:code-APN} is true, then
the $\frac{\phi(n)}{2}$ designs $\mathbb{KA}_{n,i}$ are pairwise non-isomorphic.

The codes $\mathsf{C}_2(\mathbb{KA}_{5,1})$, $\mathsf{C}_2(\mathbb{KA}_{5,2})$ and $\mathsf{C}_2(\mathbb{KA}_{7,1})$ have parameters $[32,11, 12]$, $[32,21,6]$ and $[128, 15, 56]$, respectively.
These codes are optimal. The binary code $\mathsf{C}_2(\mathbb{KA}_{7,2})$ is a self-dual linear code with parameters $[128,64,16]$.
The examples of codes above demonstrate that it is worthwhile to  study $3$-designs $\mathbb{KA}_{n,i}$ and their codes $\mathsf{C}_2(\mathbb{KA}_{n,i})$ ,
as these designs
 may yield optimal linear codes or self-dual binary codes.

\section{Summary and concluding remarks}\label{sec:contribution}

The main contributions of this paper are the following:
\begin{itemize}
\item The first one is the characterization of those base blocks supporting $3$-designs by the characteristic functions
of the base blocks,
the Walsh transforms of their characteristic functions and the number of solutions of some special equations associated with the base blocks.
\item The second is the sufficient condition for the stabilizer  of a base block under the action of the general affine group to be trivial, which was used to determine the parameters of some $3$-designs.
\item The third is the two general constructions of $3$-designs with APN functions over $\gf(2^n)$.
The first construction has produced infinite families of
      $3$-designs from Kassami APN functions over $\gf(2^n)$.
      Magma programs show that the second construction yields also $3$-designs from APN power functions over $\gf(2^n)$.
\end{itemize}

There may be other ways to select bases blocks supporting $3$-designs, a lot of work can be done
in this direction. The reader is warmly invited to attack the conjectures presented in this paper.

\noindent
{\bf Acknowledgements.} The author  thanks  Prof. Cunsheng Ding and Prof. Maosheng Xiong
for helpful discussions. The author is grateful to Prof.  Maosheng Xiong for hosting him
as a post-doctoral fellow at the Department of Mathematics, The Hong Kong University of
Science and Technology. The author acknowledges support  from  National Natural Science Foundation of China (Grant No. 11871058), 
and China West Normal University (14E013, CXTD2014-4 and the Meritocracy Research Funds).


\begin{thebibliography}{99}

\bibitem{AK92} E. F. Assmus Jr., J. D. Key. Designs and Their Codes, Cambridge University
Press, Cambridge, 1992.

\bibitem{BJL} T. Beth, D. Jungnickel, H. Lenz. Design Theory, Cambridge University
Press, Cambridge, 1999.

\bibitem{BRS67} E. R. Berlekamp, H. Rumsey, G. Solomon. On the solution of algebraic equations over finite fields. Information and control,  10(6): 553-564, 1967.

\bibitem{CM11} C. Carlet, S. Mesnager. On Dillon's class H of bent functions, Niho bent functions and o-polynomials. Journal of Combinatorial Theory, Series A,  118(8): 2392-2410, 2011.

\bibitem{Dil99} J. F. Dillon. Multiplicative difference sets via additive characters, Des. Codes Cryptogr. 17,
225-235, 1999.

\bibitem{DD04}
J.  F. Dillon, H. Dobbertin. New cyclic difference sets with Singer parameters. Finite Fields and Their Applications, 10(3): 342-389, 2004.

\bibitem{Ding15}  C. Ding. Codes from Difference Sets. World Scientific, Singapore, 2015.

\bibitem{DL17} C. Ding, C. Li. Infinite families of 2-designs and 3-designs from linear codes, Discrete Mathematics, 340,
2415-2431, 2017.

\bibitem{DT192}
C. Ding, C. Tang. Combinatorial $t$-designs from special polynomials, arXiv preprint arXiv:
	 1903.07375, 2019.

\bibitem{DY06}
C. Ding, J. Yuan.
A family of skew Hadamard difference sets, J. Combinatorial Theory Ser. A,
113, 1526--1535, 2006.

\bibitem{Dob99}  H. Dobbertin. Kasami power functions, permutation polynomials and cyclic difference sets, in: A. Pott, P.V. Kumar,
T. Helleseth, D. Jungnickel (Eds.), Difference Sets, Sequences and Their Correlation Properties, in: NATO Adv. Sci.
Inst. Ser. C Math. Phys. Sci., vol. 542, Kluwer Acad. Publ., Dordrecht,  pp. 133-158, 1999.

\bibitem{HK08} T. Helleseth, A. Kholosha. On the equation $x^{2^l+1}+x+a=0$ over $\mathrm{GF}(2^k)$. Finite Fields and Their Applications, 14(1):159-176, 2008.
\bibitem{HK10} T. Helleseth, A. Kholosha. $x^{2^l+1}+x+a=0$ and related affine polynomials over $\mathrm{GF}(2^k)$. Cryptography and Communications, 2(1):85-109, 2010.

\bibitem{KM19} K. H. Kim, S. Mesnager. Solving $ x^{2^ k+ 1}+ x+ a= 0$ in $\mathbb {F} _ {2^ n} $ with $\gcd (n, k)= 1$. arXiv preprint arXiv:1903.07481, 2019.


\bibitem{LD} H. Liu, C. Ding. Infinite families of $2$-designs from $\mathrm{GA}_1(q)$ actions. arXiv preprint arXiv:1707.02003, 2017.

\bibitem{Sihem16} S. Mesnager. Bent functions: fundamentals and results. Springer, 2016.

\bibitem{Ton98} V. D. Tonchev. Codes and designs. In: Pless V.S., Huffman W.C. (eds.), Handbook of Coding Theory, vol. II, pp. 1229--1268. Elsevier, Amsterdam, 1998.

\bibitem{Ton07}  V. D. Tonchev. Codes. In: Colbourn C.J., Dinitz J.H. (eds.), Handbook of Combinatorial Designs, 2nd edn,
pp. 677--701. CRC Press, New York, 2007.



\end{thebibliography}
\end{document}